\documentclass{llncs}
\usepackage[english]{babel}
\usepackage[utf8]{inputenc}

\usepackage{algorithm}
\usepackage{algorithmicx}
\usepackage{tikz}
\usetikzlibrary{patterns}

\usepackage{wrapfig}

\usepackage{url}

\usepackage{comment}

\usepackage{amsmath}
\usepackage{amsfonts}
\usepackage{amssymb}

\usepackage{algorithm,algpseudocode}

\newcommand{\eqdef}{\overset{\mbox{\rm\tiny def}}{=}}

\def\concurrent{\mathrel{co}}
\def\conflict{\mathrel{\#}}
\def\cutoffs{\mathit{coff}}

\newcommand{\move}[1]{\stackrel{#1}{\rightarrow}}
\def\petrinet{\mathcal{N}}
\newcommand{\postset}[1]{{{#1}^{\bullet}}}
\newcommand{\preset}[1]{{{}^{\bullet}{#1}}}
\def\unfolding{\mathcal{U}}
\def\prefix{\mathcal{P}}

\def\rdef#1{Definition~\ref{def:#1}}
\def\tup#1{\langle{#1}\rangle}

\def\ie{i.e.\ }
\def\config{\mathcal{C}}
\def\cut#1{\mathit{Cut}(#1)}
\def\marking#1{\mathit{Mark}(#1)}
\def\minconf#1{\lceil #1 \rceil}
\def\preconf#1{\lfloor #1 \rfloor}
\newcommand{\Adq}{\lhd}

\def\sgd{_{\mathrm{gd}}}

\def\mi#1{\mathit{#1}}

\def\model{{\mathcal N}}

\def\goal{g}
\def\Config{\mathcal{K}}

\def\useless{\texttt{useless-trs}}
\def\uselessmg#1#2{\useless\left(\textstyle\model,\goal,#1,#2\right)}
\def\fignored{\mathit{Useless}}
\def\ignored#1{\fignored(#1)}

\def\Vignored{\Delta}

\def\altconfig#1{\mathit{Alt}(#1)}

\usepackage{color}

\begin{document}
\mainmatter
\title{Goal-Driven Unfolding of Petri Nets}

\author{Thomas Chatain\inst{1} \and Lo{\"i}c Paulev{\'e}\inst{2}}
\institute{LSV, ENS Cachan, INRIA, CNRS, Universit\'e Paris-Saclay, France
\and
LRI UMR 8623, Univ. Paris-Sud, CNRS, Universit\'e Paris-Saclay, France}

\maketitle

\begin{abstract}
Unfoldings provide an efficient way to avoid the state-space explosion due to
interleavings of concurrent transitions when exploring the runs of a Petri net.
The theory of adequate orders allows one to define finite prefixes of unfoldings
which contain all the reachable markings.
In this paper we are interested in reachability of a single
given marking, called the goal. We propose an algorithm for computing a finite
prefix of the unfolding of a 1-safe Petri net that preserves all minimal
configurations reaching this goal. Our algorithm combines the unfolding technique
with on-the-fly model reduction by static analysis aiming at avoiding the
exploration of branches which are not needed for reaching the goal.
We present some experimental results.
\end{abstract}

\section{Introduction}

Analysing the possible dynamics of a concurrent system expressed as Petri nets can be eased by means
of unfoldings and their prefixes which avoid exploring redundant interleaving of transitions.

In this paper, we propose a method which combines the unfolding technique with
model reduction in order to explore efficiently and completely
the minimal configurations (partially ordered occurrences of transitions) which lead
to a given goal marking/marked place.
In particular, we aim at ignoring configurations that cannot reach the goal, but
also configurations containing transient cycles.

The goal-driven unfolding relies on calling, on the fly, an external model reduction procedure which
identifies transitions not part of any minimal configuration for the goal reachability from the current marking.
Those useless transitions are then skipped by the unfolding.

We show how model reduction can be applied to the unfolding of a safe Petri net
\(\model\) in such a way that it preserves minimal configurations. Then we
present an algorithm to construct a corresponding goal-driven finite prefix.

\begin{figure}[tb]
  \scalebox{.8}{\def\b{1.0}
\def\c{1.0}

\tikzstyle{place}=[circle,draw=black,fill=white,minimum width=5mm,inner sep=0pt]
\tikzstyle{transition}=[rectangle,fill=none,draw=black,minimum width=4mm,minimum height=2mm,inner sep=0pt]

\begin{tikzpicture}[>=stealth,shorten >=1pt,node distance=\c cm,auto]
  \node[place] (p_0) at (0*\b, 0*\c) [label=right:$p_0$]{\(\bullet\)};
  \node[place] (p'_0) at (0*\b, -2*\c) [label=right:$p'_0$]{};
  \node[place] (p_3) at (-2*\b, 0*\c) [label=left:$p_3$]{};
  \node[place] (p'_3) at (-2*\b, -2*\c) [label=left:$p'_3$]{\(\bullet\)};
  \node[place] (p_4) at (0*\b, -4*\c) [label=right:$p_4$]{};
  \node[place] (p_1) at (-2*\b, -4*\c) [label=left:$p_1$]{};
  \node[place] (p'_1) at (-.5*\b, -6*\c) [label=left:$p'_1$]{};
  \node[place] (p_2) at (2*\b, -4*\c) [label=right:$p_2$]{};
  \node[place] (p'_2) at (.5*\b, -6*\c) [label=right:$p'_2$]{};
  \node[place] (p_5) at (0*\b, -8*\c) [label=right:$p_5$]{};

  \node[transition] (t_0) at (0*\b, -1*\c) [label=right:$t_0$]{};
  \path[->] (p_0) edge (t_0);
  \path[->] (t_0) edge (p'_0);
  \path[<->] (p_3) edge (t_0);

  \node[transition] (t_3) at (-2.5*\b, -1*\c) [label=left:$t_3$]{};
  \path[->] (p'_3) edge (t_3);
  \path[->] (t_3) edge (p_3);

  \node[transition] (t'_3) at (-1.5*\b, -1*\c) [label=right:$t'_3$]{};
  \path[->] (p_3) edge (t'_3);
  \path[->] (t'_3) edge (p'_3);

  \node[transition] (t_1) at (-1*\b, -3*\c) [label=left:$t_1$]{};
  \path[->] (p'_0) edge (t_1);
  \path[->] (t_1) edge (p_1);
  \path[->] (t_1) edge (p_4);
  \path[<->] (p'_3) edge (t_1);

  \node[transition] (t_2) at (1*\b, -3*\c) [label=right:$t_2$]{};
  \path[->] (p'_0) edge (t_2);
  \path[->] (t_2) edge (p_2);
  \path[->] (t_2) edge (p_4);

  \node[transition] (t'_1) at (-.5*\b, -5*\c) [label=left:$t'_1$]{};
  \path[->] (p_4) edge (t'_1);
  \path[->] (t'_1) edge (p'_1);

  \node[transition] (t'_2) at (.5*\b, -5*\c) [label=right:$t'_2$]{};
  \path[->] (p_4) edge (t'_2);
  \path[->] (t'_2) edge (p'_2);

  \node[transition] (t''_1) at (-1*\b, -7*\c) [label=left:$t''_1$]{};
  \path[->] (p_1) edge (t''_1);
  \path[->] (p'_1) edge (t''_1);
  \path[->] (t''_1) edge (p_5);

  \node[transition] (t''_2) at (1*\b, -7*\c) [label=right:$t''_2$]{};
  \path[->] (p_2) edge (t''_2);
  \path[->] (p'_2) edge (t''_2);
  \path[->] (t''_2) edge (p_5);
\end{tikzpicture}}\hfill
  \scalebox{.8}{\def\b{1.0}
\def\c{1.0}

\tikzstyle{place}=[circle,draw=black,fill=white,minimum width=5mm,inner sep=0pt]
\tikzstyle{transition}=[rectangle,fill=none,draw=black,minimum width=4mm,minimum height=2mm,inner sep=0pt]
\tikzstyle{cotransition}=[transition,fill=white,draw=black,dashed]
\tikzstyle{skipped}=[fill=gray]

\begin{tikzpicture}[>=stealth,shorten >=1pt,node distance=\c cm,auto]
  \node[place] (p_0) at (0*\b, 2*\c) [label=right:$p_0$]{\(\bullet\)};
  \node[place] (p'_0) at (0*\b, -2*\c) [label=right:$p'_0$]{};
  \node[place] (p_31) at (-2*\b, .5*\c) [label=left:$p_3$]{};
  \node[place] (p_32) at (-2*\b, -.5*\c) [label=left:$p_3$]{};
  \node[place] (p_33) at (-4*\b, -5*\c) [label=left:$p_3$]{};
  \node[place] (p_3u) at (-4*\b, -3*\c) [label=left:$p_3$]{};
  \node[place] (p'_31) at (-2*\b, 2*\c) [label=left:$p'_3$]{\(\bullet\)};
  \node[place] (p'_32) at (-2*\b, -2*\c) [label=right:$p'_3$]{};
  \node[place] (p'_3u) at (-3.5*\b, -1*\c) [label=left:$p'_3$]{};
  \node[place] (p'_3u2) at (-4*\b, -7*\c) [label=left:$p'_3$]{};
  \node[place] (p'_33) at (-2.5*\b, -3*\c) [label=left:$p'_3$]{};
  \node[place] (p_41) at (-.5*\b, -4*\c) [label=left:$p_4$]{};
  \node[place] (p_42) at (.5*\b, -4*\c) [label=right:$p_4$]{};
  \node[place] (p_1) at (-2*\b, -4*\c) [label=left:$p_1$]{};
  \node[place] (p'_1) at (-.5*\b, -6*\c) [label=left:$p'_1$]{};
  \node[place] (p'_1u) at (2.5*\b, -6*\c) [label=right:$p'_1$]{};
  \node[place] (p_2) at (2*\b, -4*\c) [label=right:$p_2$]{};
  \node[place] (p'_2) at (.5*\b, -6*\c) [label=right:$p'_2$]{};
  \node[place] (p'_2u) at (-2.5*\b, -6*\c) [label=left:$p'_2$]{};
  \node[place] (p_51) at (-1*\b, -8*\c) [label=left:$p_5$]{};
  \node[place] (p_52) at (1*\b, -8*\c) [label=right:$p_5$]{};

  \node[transition] (t_0) at (0*\b, 0*\c) [label=right:$t_0$]{};
  \path[->] (p_0) edge (t_0);
  \path[->] (t_0) edge (p'_0);
  \path[->] (p_31) edge (t_0);
  \path[->] (t_0) edge (p_32);

  \node[transition] (t_3) at (-2*\b, 1.25*\c) [label=left:$t_3$]{};
  \path[->] (p'_31) edge (t_3);
  \path[->] (t_3) edge (p_31);

  \node[transition,skipped] (t_32) at (-3.5*\b, -4*\c) [label=left:$t_3$]{};
  \path[->] (p'_33) edge (t_32);
  \path[->] (t_32) edge (p_33);

  \node[cotransition,skipped] (t_3u) at (-3.5*\b, -2.5*\c) [label=left:$t_3$]{};
  \path[->] (p'_32) edge (t_3u);
  \path[->] (t_3u) edge (p_3u);

  \node[transition] (t'_3) at (-2*\b, -1.25*\c) [label=right:$t'_3$]{};
  \path[->] (p_32) edge (t'_3);
  \path[->] (t'_3) edge (p'_32);

  \node[cotransition] (t'_32) at (-4*\b, -6*\c) [label=left:$t'_3$]{};
  \path[->] (p_33) edge (t'_32);
  \path[->] (t'_32) edge (p'_3u2);

  \node[cotransition] (t'_3u) at (-3.5*\b, 0*\c) [label=left:$t'_3$]{};
  \path[->] (p_31) edge (t'_3u);
  \path[->] (t'_3u) edge (p'_3u);

  \node[transition] (t_1) at (-1*\b, -3*\c) [label=right:$t_1$]{};
  \path[->] (p'_0) edge (t_1);
  \path[->] (t_1) edge (p_1);
  \path[->] (t_1) edge (p_41);
  \path[->] (p'_32) edge (t_1);
  \path[->] (t_1) edge (p'_33);

  \node[transition] (t_2) at (1*\b, -3*\c) [label=right:$t_2$]{};
  \path[->] (p'_0) edge (t_2);
  \path[->] (t_2) edge (p_2);
  \path[->] (t_2) edge (p_42);

  \node[transition] (t'_1) at (-.5*\b, -5*\c) [label=left:$t'_1$]{};
  \path[->] (p_41) edge (t'_1);
  \path[->] (t'_1) edge (p'_1);

  \node[transition] (t'_2) at (.5*\b, -5*\c) [label=right:$t'_2$]{};
  \path[->] (p_42) edge (t'_2);
  \path[->] (t'_2) edge (p'_2);

  \node[transition,skipped] (t'_1u) at (2.5*\b, -5*\c) [label=right:$t'_1$]{};
  \path[->] (p_42) edge (t'_1u);
  \path[->] (t'_1u) edge (p'_1u);

  \node[transition,skipped] (t'_2u) at (-2.5*\b, -5*\c) [label=left:$t'_2$]{};
  \path[->] (p_41) edge (t'_2u);
  \path[->] (t'_2u) edge (p'_2u);

  \node[transition] (t''_1) at (-1*\b, -7*\c) [label=left:$t''_1$]{};
  \path[->] (p_1) edge (t''_1);
  \path[->] (p'_1) edge (t''_1);
  \path[->] (t''_1) edge (p_51);

  \node[transition] (t''_2) at (1*\b, -7*\c) [label=right:$t''_2$]{};
  \path[->] (p_2) edge (t''_2);
  \path[->] (p'_2) edge (t''_2);
  \path[->] (t''_2) edge (p_52);
\end{tikzpicture}}
  \caption{A safe Petri net (left) and a finite complete prefix (right) of its
  unfolding.
  Dashed events are flagged as \emph{cut-offs}: the unfolding procedure does not
  continue beyond them.
  Events in gray can be declared as useless by the reduction procedure for $\{p'_3,p_5\}$
  reachability, and can be skipped during
  the goal-driven prefix computation.}
  \label{fig:net0}
  \label{fig:net0unf}
\end{figure}
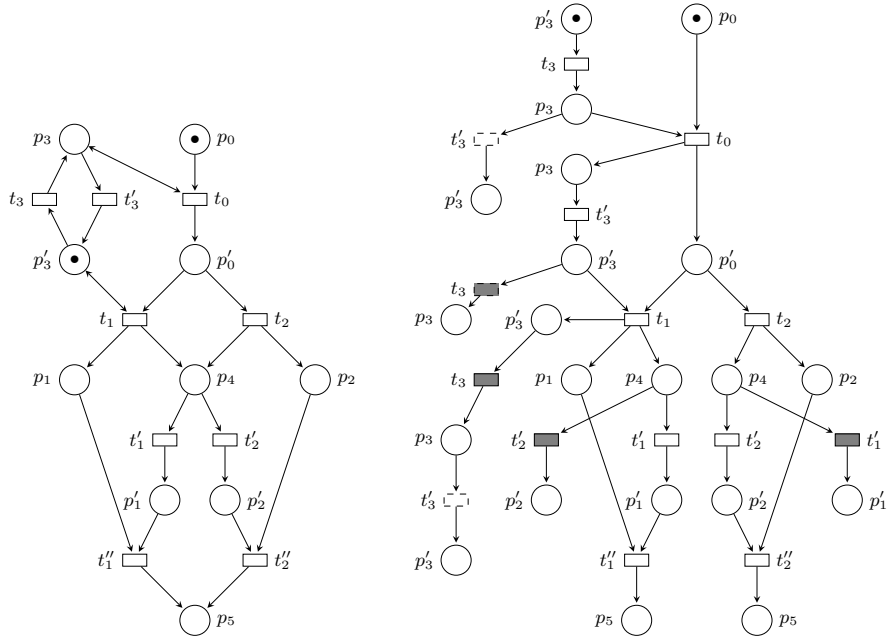

We illustrate this procedure on the Petri net of
Figure~\ref{fig:net0}. The goal is \(\{p'_3, p_5\}\). Notice that only one
occurrence of \(t_3\) is needed to reach the goal. So, after the corresponding
event, \(t_3\) can be declared useless. Also, after firing \(t_1\), \(t'_2\) is
fireable but firing it makes the goal unreachable. Therefore, a reduction
procedure may declare that \(t'_2\) is useless once \(t_1\) has occurred,
allowing one to avoid exploring this branch. Symmetrically, \(t'_1\) is useless
once \(t_2\) has occurred. It is easy to imagine a larger model where a large
piece of behaviour would be reachable from \(\{p_1, p'_2, p_3\}\) (but would not
allow to reach the goal); or from \(\{p_3,p_4\}\) (but would involve transient
cycles): the usual complete finite prefix would explore such
configurations, while our model reduction can avoid their computation.

The design of the model reduction procedure which identifies useless transitions is out of the scope of
the paper.
Instead, we consider it as a blackbox, and design our approach assuming
the reduction preserves all the minimal (acyclic) sequences of transitions leading to the goal.
Moreover, to be of practical interest, the reduction should show a
complexity lower than the reachability problem (PSPACE-complete \cite{ChengEP95}).

As detailed in Section~\ref{sec:gd-prefix}, skipping transitions declared useless by a
reduction procedure involves non-trivial modifications to the algorithm for computing the prefix of
the unfolding.
Indeed, a particular treatment of cut-offs has to be introduced in order to ensure
that the resulting goal-driven prefix includes all the minimal sequences of transitions.

The goal-driven unfolding has practical applications in systems biology
\cite{Samaga10-JCB}.
Indeed, numerous dynamical properties relevant for biological networks focus on the
reachability of the activity of a particular node in the network, typically a
transcription factor known to control a given cellular phenotype.
In this perspective, having computational methods that can be tailored for such
narrow reachability properties is of practical interest.
The \emph{completeness of the minimal sequences of transitions} for the goal reachability is
\emph{critical} for several analyses of biological system dynamics.
An example is the identification of parts of the network that play a central
role to activate a node of interest.
By altering such parts (e.g., with mutations) one can expect prevent such an
activation \cite{PAK13-CAV}.
If the analysis considers only a partial set of minimal sequences, there is no
guarantee that the predicted mutations are sufficient to prevent the goal
reachability.

\paragraph{Related Work.}
Numerous work address the computation of reachable states in concurrent systems
using unfoldings.
\cite{EsparzaS01} compares several algorithms for checking reachability based on
a previously computed finite complete prefix of a Petri net.
\cite{BaldanCK01} defines over-approximations of the
unfolding (i.e., which contains all the reachable markings, but potentially
more) for graph transformation systems.

Despite the negative result \cite{EsparzaKS08} which states that
depth-first-search strategies are not correct for classical unfolding
algorithms, \cite{BonetHHT08} defines \emph{directed unfolding} of Petri nets,
which is closely related to our goal-driven unfolding.
They rely on a heuristic function (on configurations) to generate an
ordering of the events for making a given transition appear as soon as possible
during the unfolding.
In addition, they can consider heuristic functions to detect
configuration from which the goal transition is not reachable.
In such a case, no extension will be made to that configuration, which may
significantly prune the computed prefix.
The major difference with the work presented in the paper is that directed
unfolding does not prune transitions
leading to
spurious transient cycles on the way to the goal. Actually, in their terms,
our reduction procedure would not be considered \emph{safely pruning} because we
discard (non-minimal) configurations reaching the goal.
In a sense, the reduction they achieve on the prefix size corresponds to the
extreme case when our external reduction procedure returns the full model
if the goal is reachable, and the empty model if not.
Indeed, except for the case when the goal is detected as non-reachable, all the
other configurations are kept in the directed unfolding, whereas our approach can
potentially output a prefix containing only, but all, minimal configurations for
the goal reachability.

Less related to our work, static analysis techniques were also used in
combination with partial order reductions.
\cite{FlanaganG05,WangYKG08} rely on an on-the-fly detection of independence
relations by static analysis, to improve partial order reductions.

\paragraph{Outline.}
Section~\ref{sec:unf} gives the basics of Petri net unfoldings and of their complete finite prefixes.
The concepts of minimal configuration and model reduction are introduced in
Section~\ref{sec:reduction}, and
Section~\ref{sec:godunf} details the goal-driven unfolding and prefix with proofs of completeness.
Finally, Section~\ref{sec:experiments} applies the goal-driven prefix to actual
biological models, and Section~\ref{sec:conclusion} concludes this paper.


\section{Unfoldings of Petri nets}
\label{sec:unf}

In this section, we explain the basics of Petri net unfoldings.
A more extensive treatment of the
theory explained here can be found, e.g., in \cite{Esparza08}.
Roughly speaking, the unfolding of
a Petri net $\petrinet$ is an ``acyclic'' Petri net $\unfolding$ that has the
same behaviours as $\petrinet$ (modulo homomorphism). In general, $\unfolding$
is an infinite net, but if $\petrinet$ is safe, then it is possible
\cite{McM92} to compute a finite prefix $\prefix$
of $\unfolding$ that is ``complete''
in the sense that every reachable marking of $\petrinet$ has a reachable
counterpart in $\prefix$. Thus, $\prefix$ represents the set of reachable
markings of $\petrinet$.
Figure~\ref{fig:net0} shows a Petri net and a finite complete prefix of its
unfolding.

We now give some technical definitions to introduce unfoldings formally.


\begin{definition}[(Safe) Petri Net]
  A \emph{(safe) Petri net} is a tuple $\petrinet=\tup{P, T, F, M_0}$ where $P$
  and $T$ are sets of \emph{nodes} (called \emph{places} and \emph{transitions}
  respectively), and $F\subseteq(P\times T)\cup(T\times P)$ is a \emph{flow
    relation} (whose elements are called \emph{arcs}). A subset $M\subseteq P$
  of the places is called a marking, and $M_0$ is a distinguished \emph{initial
    marking}.
\end{definition}
For any node $x\in P\cup T$, we call \emph{pre-set} of $x$ the set
$\preset{x}=\{y\in P\cup T\mid (y,x)\in F\}$ and \emph{post-set} of $x$ the set
$\postset{x}=\{y\in P\cup T\mid (x,y)\in F\}$.
These notations are extended to sets $Y \subseteq P\cup T$, with
$\preset Y = \cup_{x\in Y}\preset x$ and $\postset Y=\cup_{x\in Y}\postset x$.

A transition $t\in T$ is \emph{enabled} at a marking $M$ if and only if $\preset{t}\subseteq M$.
Then $t$ can \emph{fire}, leading to the new marking $M'=(M\setminus\preset{t})\cup\postset{t}$.
We write $M\move{t}M'$.
A \emph{firing sequence} is a (finite or infinite) word $w=t_1t_2\dots$ over $T$
such that there exist markings $M_1,M_2,\dots$ such that
$M_0\move{t_1}M_1\move{t_2}M_2\dots$\@
For any such firing sequence $w$, the markings $M_1,M_2,\dots$ are called \emph{reachable markings}.

The Petri nets we consider are said to be \emph{safe} because we will assume
that any reachable marking $M$ is such that for any $t\in T$ that can fire from
$M$ leading to $M'$, the following property holds: $\forall p\in M\cap M'$,
$p\in\preset{t}\cap\postset{t} \vee p\notin\preset{t}\cup\postset{t}$.

Figure~\ref{fig:net0} (left) shows an example of a safe Petri net.
The places are represented by circles and the transitions by rectangles (each one with a label identifying it).
The arrows represent the arcs.
The initial marking is represented by dots (or tokens) in the marked places.


\begin{definition}[Causality, conflict, concurrency]
\label{def:causality}
Let $\petrinet=\tup{P,T,F,M_0}$ be a net and $t,t'\in T$ two transitions
of $\petrinet$.
We say that
$t$ is a \emph{causal predecessor} of $t'$, noted $t<t'$, if there exists a
non-empty path of arcs from $t$ to $t'$. We note $t\leq t'$ if $t<t'$ or $t=t'$.
If $t\leq t'$ or $t'\leq t$, then $t$ and $t'$ are said to be
\emph{causally related}.
The set of causal predecessors of \(t\) is denoted \(\preconf{t}\). We write
\(\minconf{t}\) for \(\preconf{t} \cup \{t\}\), which we call the \emph{causal
  past} of \(t\).
Transitions $t$ and $t'$ are \emph{in conflict}, noted $t\conflict t'$, if there
exist
$u,v\in T$ such that \(u \neq v\), $u\le t$, $v\le t'$ and
$\preset{u}\cap\preset{v}\ne \emptyset$. We call
$t$ and $t'$ \emph{concurrent}, noted $t\concurrent t'$, if they are neither
causally related nor in conflict.
\end{definition}

As we said before, an unfolding is an ``acyclic'' net.
This notion of acyclicity is captured by \rdef{occurrencenet}.
As is convention in the unfolding literature, we shall refer to the
places of an occurrence net as \emph{conditions} and to its transitions
as \emph{events}. Due to the structural constraints, the firing
sequences of occurrence nets have special properties: if some condition $c$ is
marked during a run, then the token on $c$ was either present initially
or produced by one particular event (the single event in $\preset{c}$);
moreover, once the token on $c$ is consumed, it can never be replaced by
another token, due to the acyclicity constraint on $<$.

\begin{definition}[Occurrence net]
\label{def:occurrencenet}
An \emph{occurrence net}
$\mathcal{O}=\tup{P,T,G,M_0}$
is a Petri net
$\tup{P,T,F,M_0}$
with $P=C$, $T=E$, $F=G$, $M_0=C_0$
for which:
\begin{enumerate}
\item The causality relation $<$ is acyclic;
\item $|\preset{p}|\le 1$ for all places $p$,
  and $p \in M_0$ iff $|\preset{p}|=0$;
\item for every transition $t$, $t\conflict t$ does not hold, and \(\{x \mid x
  \leq t\}\) is finite.
\end{enumerate}
\end{definition}


\begin{definition}[Configuration, cut]
\label{def:configuration}
Let $\mathcal{O}=\tup{C,E,G,C_0}$ be an occurrence net.
A set $\config\subseteq E$ is called
\emph{configuration} (or \emph{process}) of $\mathcal{O}$ if (i) $\config$ is \emph{causally
closed}, \ie for all $e,e'\in E$ with $e'<e$, if $e\in\config$ then
$e'\in\config$; and (ii) $\config$ is conflict-free, \ie if $e,e'\in\config$,
then $\neg(e\conflict e')$.
The \emph{cut} of $\config$, denoted $\cut{\config}$, is the set of
conditions $(C_0\cup\postset{\config})\setminus\preset{\config}$.
\end{definition}

An occurrence net $\mathcal O$ with a net homomorphism $h$ mapping
its conditions and events to places and transitions of a net \(\model\)
is called a \emph{branching process} of \(\model\).
Intuitively, a configuration of \(\mathcal{O}\) is a set of events that can fire
during a firing sequence of $\petrinet$, and its cut is the set of conditions
marked after that sequence.

\paragraph{Unfolding.}
Let $\petrinet=\tup{P,T,F,M_0}$
be a safe Petri net. The unfolding $\unfolding=\tup{C,E,G,C_0}$ of
$\petrinet$ is the unique (up to isomorphism) maximal branching process
such that
the firing sequences and reachable markings of $\unfolding$ represent exactly
the firing sequences and reachable markings of $\petrinet$ (modulo $h$).
$\unfolding$ is generally infinite but its conditions and events can be
inductively constructed as follows:
\begin{enumerate}
\item The conditions $C$ are a subset of $(E\cup\{\bot\})\times P$.
 For a condition $c=\tup{x,p}$, we will have $x=\bot$ iff $c\in C_0$;
 otherwise $x$ is the singleton event in $\preset{c}$. Moreover, $h(c)=p$.
 The initial marking $C_0$ contains one condition $\tup{\bot,p}$
 per initially marked place $p$ of $\petrinet$.
\item The events $E$ are a subset of $2^C\times T$.
 More precisely, we have an event $e=\tup{C',t}$ for every set $C'\subseteq C$
 such that $c\concurrent c'$ holds for all $c,c'\in C'$ and
 $\{\,h(c)\mid c\in C'\,\}=\preset{t}$. In this case, we add edges $\tup{c,e}$
 for each $c\in C'$ (\ie $\preset{e}=C'$), we set $h(e)=t$, and for each
 $p\in\postset{t}$, we add to $C$ a condition $c=\tup{e,p}$, connected by
 an edge $\tup{e,c}$.
\end{enumerate}
Intuitively, a condition $\tup{x,p}$ represents the possibility of putting
a token onto place $p$ through a particular firing sequence, while an event
$\tup{C',t}$ represents a possibility of firing transition $t$ in a particular
context.

Every firing sequence \(\sigma\) is represented by a configuration of
$\unfolding$; we denote this configuration \(\Config(\sigma)\).
Conversely, every configuration $\config$ of $\unfolding$ represents one or
several firing sequences (\(\Config\) is not injective in general); these
firing sequences are equivalent up to permutation of concurrent
transitions. Their (common) resulting marking corresponds, due to the
construction of $\unfolding$, to a reachable marking of $\petrinet$. This
marking is defined as $\marking{\config}:=\{\,h(c)\mid c\in\cut{\config}\,\}$.

\paragraph{Finite Complete Prefix.}
The unfolding $\unfolding$ of a finite safe Petri net \(\model\) is infinite in
general, but it shows some regularity because \(\model\) has finitely many
markings and two events \(e\) and \(e'\) having \(\marking{\minconf{e}} =
\marking{\minconf{e'}}\) have isomorphic extensions.



It is known \cite{McM92,ERV02} that one can construct a \emph{finite complete
  prefix} \(\prefix\) of \(\unfolding\), i.e.\ a causally closed set \(E'\) of
events of \(\unfolding\) which is sufficiently large for satisfying the
following: for every reachable marking $M$ of $\petrinet$ there exists a
configuration $\config$ of $\prefix$ such that $\marking{\config}=M$. One can
even require that for each transition $t$ of \(\model\) enabled in $M$, there is
an event $\tup{C,t}\in E'$ enabled in $\cut{\config}$.

The idea of the construction is to explore the future of only one among the
events \(e\) having equal \(\marking{\minconf{e}}\). The selected event is the
one having minimal \(\minconf{e}\) w.r.t.\ a so-called \emph{adequate order} on
the finite configurations of \(\unfolding\). The others are flagged
as \emph{cut-offs}; they do not ``contribute any new reachable markings''. These
events are represented by dashed lines in Figure~\ref{fig:net0unf}.





\begin{definition}[Adequate orders]\label{adequate}
A strict partial order $\Adq$ on the finite configurations of the
unfolding of a safe Petri net $\model$ is called \emph{adequate} if:
\begin{itemize}
\item it refines (strict) set inclusion $\subsetneq$, \ie $C\subsetneq C'$
  implies $C\Adq C'$, and
\item it is preserved by finite extensions, i.e.\ for every pair of
  configurations $C$, $C'$ such that $\marking{C}=\marking{C'}$ and $C\Adq C'$,
  and for every finite extension $D$ of $C$, the finite extension $D'$ of $C'$
  which is isomorphic to $D$ satisfies $C \uplus D \Adq C'\uplus D'$.
\end{itemize}
\end{definition}
The initial definition of adequate orders \cite{ERV02} also requires that \(\Adq\)
is well founded, but \cite{ChatainK07} showed that, for unfoldings of safe Petri
nets, well-foundedness is a consequence of the other requirements.



Efficient tools \cite{mole,punf} exist for computing finite complete prefixes.


\section{Goal-Oriented Model Reduction}
\label{sec:reduction}

The goal-driven unfolding relies on model reduction procedures which preserve
minimal firing sequence to reach
a given goal \(\goal\). 
These reductions aim at removing as many
transitions as possible among those that do not participate in any minimal
firing sequence.
This section details the properties required by our method and introduce several
notations used in the rest of the paper.


\begin{definition}[Minimal firing sequence]
  A firing sequence \(t_1 \dots t_n\) of a Petri net \(\model = \tup{P, T, F,
    M_0}\) visiting markings $M_0\move{t_1}M_1\move{t_2}M_2\dots\move{t_n}M_n$
  is said \emph{cycling} if it visits twice the same marking, i.e.\ \(M_i =
  M_j\) for some \(0 \leq i < j \leq n\).
  A \emph{minimal firing sequence} of \(\model\) to a goal \(\goal\)
  is a firing sequence \(t_1 \dots t_n\) leading to \(\goal\)
  which has no feasible permutation\footnote{Contrary to what is common in concurrency
    theory, we do not necessarily restrict to permutations of independent
    transitions w.r.t.\ an independence relation.} being a cycling firing
  sequence of \(\model\).
\end{definition}

For example with Petri net of Fig.~\ref{fig:net0} and considering the goal
\(\{p'_3, p_5\}\),
\(t_3 t_0 t'_3 t_2 t_3
t'_2 t''_2 t'_3\) is not minimal because its permutation \(t_3 t_0 t'_3 t_3 t_2
t'_2 t''_2 t'_3\) is also feasible and visits the marking \(\{p_3, p'_0\}\) twice.
Intuitively, the cycle \(t'_3 t_3\) can be removed. The minimal firing sequences
of \(\model\) to the goal are \(t_3 t_0 t'_3 t_1 t'_1 t''_1\),
\(t_3 t_0 t_2 t'_3 t'_2 t''_2\) and their feasible permutations, for instance
\(t_3 t_0 t_2 t'_2 t''_2 t'_3\).

\begin{remark}
  Alternatively, the goal can be seen not as a marking but simply as a set of
  places to be marked together, possibly with others. Then, one is looking
  for sequences reaching any marking \(M\) with \(\goal \subseteq M\).
  For minimality, we would then require additionally that no intermediate
  marking reached before the end of the sequence marks the places in \(\goal\)
  (and the same for its permutations).
\end{remark}




\begin{definition}[Minimal configuration]
  A \emph{minimal configuration} of a Petri net \(\model\) to a goal \(\goal\) is a
  configuration \(E = \Config(\sigma)\) for some minimal firing sequence
  \(\sigma\) of \(\model\) to \(\goal\). Notice that, since all the other
  \(\sigma'\) such that \(E = \Config(\sigma')\) are permutations of \(\sigma\),
  they are all minimal.
\end{definition}

\begin{lemma}\label{lem:minimal_preserves_reach}
  The goal \(\goal\) is reachable iff it is reachable by a minimal firing
  sequence (and, consequently, by a minimal configuration).
\end{lemma}
\begin{proof}
  Assume that \(\goal\) is reachable by a non-minimal firing sequence
  \(\sigma\). This means that \(\sigma\) has a permutation \(t_1\dots t_n\)
  which visits the same marking twice, i.e.\
  $M_0\move{t_1}M_1\move{t_2}M_2\dots\move{t_i}M_i\dots\move{t_j}M_j\dots\move{t_n}M_n = \goal$
  with \(M_i = M_j\) and \(i < j\). Then \(\goal\) is also reachable by the
  strictly shorter sequence \(t_1 \dots t_i t_{j+1} \dots t_n\).
  This operation can be iterated if needed; it always terminates and gives a
  minimal firing sequence which reaches the goal \(\goal\).
\end{proof}


\begin{definition}[Reduction procedure, useless transitions]\label{def:reduction}
  A \emph{reduction procedure} \(\useless\) is a function which outputs, for
  a safe Petri net $\model$ and a goal $\goal\subseteq P$, a set
  \(\useless(\model, \goal) \subseteq T\) of transitions of \(\model\) which do
  not occur in any \emph{minimal} firing sequence of $\model$ to goal $\goal$:
  for every minimal firing sequence \(t_1\dots t_n\) to goal \(\goal\),
  \(\useless(\model, \goal) \cap \{t_1, \dots, t_n\} = \emptyset\).
\end{definition}

For example, let \(\model = \tup{P, T, F, M_0}\) be the Petri net of
Fig.~\ref{fig:net0}.
All the transitions occur in at least one minimal firing sequence to the goal
\(\goal = \{p'_3, p_5\}\),
so every reduction procedure outputs \(\useless(\model, \goal)=\emptyset\).
After firing \(t_3 t_0 t'_3\), one reaches marking \(\{p'_3,
p'_0\}\) from which the only minimal firing sequences to $\goal$ are \(t_1 t'_1
t''_1\) and \(t_2 t'_2 t''_2\). Hence, a reduction procedure called as
\(\useless\) \((\tup{P, T, F, \{p'_3, p'_0\}}, \goal)\) may declare \(t_0\), \(t_3\)
and \(t'_3\) useless, or any subset of those.

Given a Petri net $\model=\tup{P,T,F,M_0}$, 
$\model\setminus\useless(\model,\goal)$ denotes the reduced model
$\tup{P,T',F',M_0}$ where
  \(T' = T \setminus \useless(\model,\goal)\)
  and
  \(F' = F \cap ((P \times T') \cup (T' \times P))\).
  Property~\ref{pty:reduction} derives from Def.~\ref{def:reduction} and
  Lemma~\ref{lem:minimal_preserves_reach}.
\begin{property}\label{pty:reduction}
  Every reduction procedure preserves reachability of the goal: \(\goal\)
  is reachable in \(\model\) iff it is reachable in
  $\model\setminus\useless(\model,\goal)$.
\end{property}


In the sequel, we aim at iterating the reduction procedures: starting from a
model \(\model = \tup{P, T, F, M_0}\) and a goal \(\goal\), we will apply the
reduction to \(\model\), then explore the reduced net
$\model\setminus \useless(\model,\goal)$;
later on, we will apply again the reduction from a reached state \(M\) and
compute \(\useless(\model', \goal)\) with
\(\model' = \tup{P, T, F, M}\setminus\useless(\model,\goal)\) allowing to explore a further reduced
net
$\model'\setminus\useless(\model',\goal)$ from $M$.
These iterated calls to the reduction procedure are justified by the following
lemma.
\begin{lemma}
   Any minimal sequence in $\model\setminus\useless(\model,\goal)$
   is minimal in  $\model$.
\end{lemma}
\begin{proof}
Any firing sequence of $\model\setminus\useless(\model,\goal)$ is a firing
sequence of $\model$, and the minimality criterion does not depend
on the set of transitions in $\model$.
\end{proof}

In the remainder of the paper, for a Petri net $\model=\tup{P,T,F,M_0}$ and any set
$I\subseteq T$ and reachable marking $M$,
we write \(\uselessmg{M}{I}\) for \(\useless(\tup{P,T,F,M}\setminus I, \goal) \cup I\).



\section{Goal-Driven Unfolding}
\label{sec:godunf}

In this section, we first show that model reduction can be performed during the
unfolding of a safe Petri net $\model$ while preserving the minimal
configurations to the goal.
Next we present an algorithm to construct a finite goal-driven prefix which
preserves the reachable markings of the goal-driven unfolding.

\subsection{Guiding the Unfolding by a Model Reduction Procedure}
\label{sec:gd-unf}

The principle of the goal-driven unfolding is that, for some events \(e\) in the
unfolding (at discretion), a model reduction procedure \(\useless\) is called
and the transitions declared useless will not be considered in the future of
\(e\). More precisely, the reduction procedure is called on the marking
\(\marking{\minconf{e}}\) of the causal past of \(e\).

Notice that the reduction procedure may already have been used on some events in
the causal past of \(e\). Then,
\begin{itemize}
\item even if \(\useless\) is not called on \(e\), information about useless
  transitions inherited from the causal predecessors of \(e\) can be used (without calling the model
  reduction procedure), and this will already prune some branches in the future of \(e\);
\item if the reduction procedure is called on \(\marking{\minconf{e}}\),
  it can take as input the model already reduced by the transitions declared
  useless after some event in the causal past of \(e\).
\end{itemize}

Let $\unfolding = \tup{C,E,G,C_0}$ be the full unfolding of a safe Petri net
\(\model\).
Denote \(E'\) the set of events on which the reduction procedure is called.
The set \(E'\) and the reduction procedure define the set of transitions
\(\ignored e\) to be ignored in the future of an event \(e \in E\).
We define \(\fignored\) inductively as:
\[\ignored{e} \eqdef \left\{
\begin{array}{@{}l@{\quad}l@{}}
  \bigcup_{e' \in \preconf{e}} \ignored{e'} & \mbox{if \(e \notin E'\)} \\
  \uselessmg
      {\marking{\minconf{e}}}
      {\bigcup_{e' \in \preconf{e}} \ignored{e'}} & \mbox{if \(e \in E'\).}
\end{array}\right.
\]
Thus, every event \(e = \tup{C, t} \in E\) such that \(t \in
\ignored{e'}\) for some \(e' \in \preconf{e}\), is discarded from the goal-driven
unfolding. Denote \(E_\mi{Ignored}\) the set of such events.

It remains to define the goal-driven unfolding as the maximal prefix of the full
unfolding \(\unfolding\) having no event in \(E_\mi{Ignored}\). Since every
discarded event automatically discards all its causal successors, the set of
events remaining in the goal-driven unfolding \(\unfolding_{\mathrm{gd}}\) is
\[E_\mathrm{gd} \eqdef \{e \in E \mid \minconf{e} \cap E_\mi{Ignored} = \emptyset\}\;.\]

Notice that the events and conditions of the goal-driven unfolding as defined
above can be constructed inductively following the procedure described in
Section~\ref{sec:unf}, enriched so that it attaches the set \(\ignored{e}\) to
every new event \(e\).

\begin{theorem}\label{th:godunf}
(proof in Appendix~\ref{proof:godunf})
The goal-driven unfolding preserves all minimal configurations from \(M_0\) to
the goal.
\end{theorem}

A direct corollary is that the goal is reachable in \(\model\) iff the
goal-driven unfolding contains a configuration which reaches it.

Notice that the precise definition of minimal sequences/configurations is
crucial here, and especially the fact that the reduction procedure preserves
\emph{all} minimal sequences/configurations. Indeed, imagine a situation where
the minimal firing sequences to the goal fire two concurrent transitions \(t_1\)
and \(t_2\) and then one out of two possible transitions \(t_3\) and \(t_4\). A
reduction procedure which would guarantee only the preservation of \emph{some}
minimal firing sequence to the goal could declare \(t_3\) useless when called
after the event corresponding to \(t_1\), and declare \(t_4\) useless when
called after \(t_2\), thus preventing to reach the goal.


\subsection{Goal-Driven Prefix}
\label{sec:gd-prefix}


We now define a finite goal-driven prefix. Our Algorithm~\ref{alg:gdprefix}
relies on the theory of adequate orders~\cite{ERV02} developed for unfoldings.
Any adequate order on the configurations of the full unfolding can be used, but,
since our goal-driven unfolding prunes some branches of the unfolding, we have
to adapt the construction.

A prefix $\prefix$ has the same structure as an unfolding, with an additional field $\cutoffs$ for the set of cut-off events.
As usual, the procedure \(\Call{Putative-GD-Prefix}{}\) extends iteratively the prefix \(\prefix=\tup{C,E,G,C_0,\cutoffs}\).
An extension is an event $e=\tup{C',t}$ with $C'\subseteq C$ s.t.
$\forall c,c'\in C',\ c \concurrent c'$,
$\{h(c)\mid c\in C'\} = \preset t$, and
$\forall \tup{e',p}\in C',\ e'\notin \cutoffs$.
Here the procedure maintains a map $\Vignored$ of
transitions that can be ignored, and considers an extension $e=\tup{C',t}$
only if the transition \(t\) is not declared useless, i.e., $t$ is absent from
$\Vignored(c')$ for all pre-condition $c'\in C'$.
\medskip

\begin{algorithm}[t]
\caption{\label{alg:gdprefix}Algorithm for goal-driven prefix computation.
}

\algrenewcommand{\algorithmiccomment}[1]{\hskip1em /\textit{#1}/}
\begin{algorithmic}[1]
\Procedure{Putative-GD-Prefix}{$\petrinet,\Vignored$}
with $\petrinet=\tup{P,T,F,M_0}$
\State $\prefix \gets \tup{
	C\gets \{\tup{\bot,p}\mid p\in M_0\},
	E\gets \emptyset,
	G\gets \emptyset,
	C_0\gets\{\tup{\bot,p}\mid p\in M_0\},
	\cutoffs\gets \emptyset}$
\Repeat
\State 
Let $e=\tup{C',t}$ be a 
$\Adq$-minimal extension of \(\prefix\) s.t.\
$t\notin \bigcup_{c'\in C'} \Vignored(c')$.
\label{line:extension}
\State $E\gets E\cup\{e\}$
\State $C\gets C\cup\{\tup{e,p}\mid p\in\postset t\}$
\State $G\gets G\cup\{\tup{c',e}\mid c'\in C'\}\cup\{\tup{e,\tup{e,p}}\mid p\in\postset t\}$
\If{$\exists e'\in E$ s.t. $\marking{\minconf e}=\marking{\minconf{e'}}$}
\State $\cutoffs \gets \cutoffs\cup\{e\}$ \label{line:cut-off}
\Comment{$e$ is a cut-off event}
\EndIf
\ForAll{$c\in \{\tup{e,p}\mid p\in\postset t\}$ s.t. $c\notin\Vignored$}
\Comment{extend $\Vignored$ with new cond.}
\label{line:extend-domain}
\State $\Vignored(c) \gets \ignored{c,\Vignored,\prefix}$
\EndFor
\Until{no extension exists}
\EndProcedure
\Statex
\Procedure{Post-$\Vignored$}{$\Vignored,\prefix$} with \(\prefix = \tup{C, E, G, C_0, \cutoffs}\)
\State $\Vignored' \gets\Vignored$
\Comment{copy map $\Vignored$}
\For{$e\in E$ following $\Adq$ order}
\ForAll{$c\in \postset e$}
\State 
$\Vignored'(c) \gets \Vignored'(c) \cap \ignored{c,\Vignored,\prefix}$
\label{line:post-useless}
\EndFor
\If{$\exists e' \in E\setminus\cutoffs$ s.t. $\marking{\minconf e}=\marking{\minconf{e'}}$}
\label{line:post-cutoffs}
	\ForAll{$c'\in\cut{\minconf{e'}}$ with $c\in\cut{\minconf{e}}$ and $h(c)=h(c')$}
	\State $\Vignored'(c')\gets \Vignored'(c')\cap\Vignored'(c)$
	\EndFor
\EndIf
\EndFor
\EndProcedure
\Statex
\Procedure{GD-Prefix}{$\petrinet$}
with $\petrinet=\tup{P,T,F,M_0}$
\State $\Vignored' \gets\{ \tup{\bot,p} \mapsto \emptyset\mid p\in M_0 \}$
\Repeat
\State $\Vignored\gets \Vignored'$
\Comment{copy map $\Vignored'$}
\State $\prefix \gets \Call{Putative-GD-Prefix}{\petrinet,\Vignored}$
\Comment{can add new entries in $\Vignored$}
\State $\Vignored'\gets $\Call{Post-$\Vignored$}{$\Vignored,\prefix$}
\label{line:recompute}
\Until{$\Vignored'=\Vignored$}
\EndProcedure
\end{algorithmic}
\end{algorithm}

The difficult part is that, when
an event \(e\) is declared cut-off because \(\marking{\minconf{e}} =
\marking{\minconf{e'}}\) for an event \(e' \Adq e\), nothing guarantees that the
transitions allowed after \(e\) are also allowed after \(e'\). Then, \(e\) and
\(e'\) have the same future in the full unfolding, but not necessarily in the
goal-driven unfolding.

\begin{figure}[tb]
  \scalebox{.8}{\def\b{1.2}
\def\c{1.5}

\tikzstyle{place}=[circle,draw=black,fill=white,minimum width=5mm,inner sep=0pt]
\tikzstyle{transition}=[rectangle,fill=none,draw=black,minimum width=4mm,minimum height=2mm,inner sep=0pt]

\begin{tikzpicture}[>=stealth,shorten >=1pt,node distance=\c cm,auto]
  \node[place] (p_0) at (0*\b, 0*\c) [label=right:$p_0$]{\(\bullet\)};
  \node[place] (p_1) at (-1*\b, -2*\c) [label=above:$p_1$]{};
  \node[place] (p_2) at (0*\b, -2*\c) [label=below:$p_2$]{};
  \node[place] (p_3) at (2*\b, -2*\c) [label=below:$p_3$]{};
  \node[place] (p_4) at (-1*\b, -4*\c) [label=right:$p_4$]{};

  \node[transition] (a) at (0*\b, -1*\c) [label=right:$a$]{};
  \path[->] (p_0) edge (a);
  \path[->] (a) edge (p_1);
  \path[->] (a) edge (p_2);

  \node[transition] (a') at (1*\b, -1*\c) [label=right:$a'$]{};
  \path[->] (p_0) edge (a');
  \path[->] (a') edge (p_1);
  \path[->] (a') edge (p_3);

  \node[transition] (b) at (1*\b, -1.7*\c) [label=above:$b$]{};
  \path[->] (p_2) edge (b);
  \path[->] (b) edge (p_3);

  \node[transition] (b') at (1*\b, -2.3*\c) [label=below:$b'$]{};
  \path[->] (p_3) edge (b');
  \path[->] (b') edge (p_2);

  \node[transition] (c) at (-1*\b, -3*\c) [label=right:$c$]{};
  \path[->] (p_1) edge (c);
  \path[->] (c) edge (p_4);
  \path[<->] (c) edge (p_2);

\end{tikzpicture}}\hfill
  \scalebox{.8}{\def\b{1}
\def\c{.8}

\tikzstyle{place}=[circle,draw=black,fill=none,minimum width=5mm,inner sep=0pt]
\tikzstyle{transition}=[rectangle,fill=none,draw=black,minimum width=4mm,minimum height=2mm,inner sep=0pt]
\tikzstyle{cotransition}=[transition,fill=none,draw=black,dashed]

\begin{tikzpicture}[>=stealth,shorten >=1pt,node distance=\c cm,auto]
  \fill[fill=lightgray, opacity=0.5] (-3*\b, -2*\c) -- (-1*\b, -4*\c) -- (-.5*\b,-10.5*\c) -- (-3.5*\b,-10.5*\c) -- (-3*\b, -2*\c);
  \fill[fill=lightgray, opacity=0.5] (1*\b, -2*\c) -- (3*\b, -2*\c) -- (3.5*\b,-8.5*\c) -- (.5*\b,-8.5*\c) -- (1*\b, -2*\c);
  \node[place] (p_0) at (0*\b, 0*\c) [label=right:$p_0$]{\(\bullet\)};
  \node[place] (p_1) at (-3*\b, -2*\c) [label=left:$p_1$]{};
  \node[place] (p_1') at (1*\b, -2*\c) [label=left:$p_1$]{};
  \node[place] (p_2) at (-1*\b, -2*\c) [label=left:$p_2$]{};
  \node[place] (p_2') at (3*\b, -4*\c) [label=left:$p_2$]{};
  \node[place] (p_2'2) at (3*\b, -6*\c) [label=left:$p_2$]{};
  \node[place] (p_22) at (-1*\b, -6*\c) [label=left:$p_2$]{};
  \node[place] (p_23) at (-1*\b, -8*\c) [label=left:$p_2$]{};
  \node[place] (p_3) at (-1*\b, -4*\c) [label=right:$p_3$]{};
  \node[place] (p_32) at (-1*\b, -10*\c) [label=right:$p_3$]{};
  \node[place] (p_3') at (3*\b, -2*\c) [label=right:$p_3$]{};
  \node[place] (p_3'2) at (3*\b, -8*\c) [label=right:$p_3$]{};
  \node[place] (p_4) at (-3*\b, -8*\c) [label=right:$p_4$]{};
  \node[place] (p_4') at (1*\b, -6*\c) [label=right:$p_4$]{};

  \node[transition] (a) at (-2*\b, -1*\c) [label=right:$a$]{};
  \path[->] (p_0) edge (a);
  \path[->] (a) edge (p_1);
  \path[->] (a) edge (p_2);

  \node[cotransition] (a') at (2*\b, -1*\c) [label=right:$a'$, label=above:$(e)$]{};
  \path[->] (p_0) edge (a');
  \path[->] (a') edge (p_1');
  \path[->] (a') edge (p_3');

  \node[transition] (b) at (-1*\b, -3*\c) [label=left:$b$, label=right:$(e')$]{};
  \path[->] (p_2) edge (b);
  \path[->] (b) edge (p_3);

  \node[transition] (b1) at (3*\b, -7*\c) [label=left:$b$]{};
  \path[->] (p_2'2) edge (b1);
  \path[->] (b1) edge (p_3'2);

  \node[transition] (b2) at (-1*\b, -9*\c) [label=left:$b$]{};
  \path[->] (p_23) edge (b2);
  \path[->] (b2) edge (p_32);

  \node[transition] (b') at (3*\b, -3*\c) [label=left:$b'$]{};
  \path[->] (p_3') edge (b');
  \path[->] (b') edge (p_2');

  \node[transition,pattern=north east lines] (b'2) at (-1*\b, -5*\c) [label=left:$b'$]{};
  \path[->] (p_3) edge (b'2);
  \path[->] (b'2) edge (p_22);

  \node[transition] (c) at (-2*\b, -7*\c) [label=right:$c$]{};
  \path[->] (p_1) edge (c);
  \path[->] (p_22) edge (c);
  \path[->] (c) edge (p_4);
  \path[->] (c) edge (p_23);

  \node[transition] (c') at (2*\b, -5*\c) [label=right:$c$]{};
  \path[->] (p_1') edge (c');
  \path[->] (p_2') edge (c');
  \path[->] (c') edge (p_4');
  \path[->] (c') edge (p_2'2);

  \path[->, bend right, dashed] (a') edge (b) {};

  \node[transition] (c'') at (-4*\b, -3.5*\c) [label=right:$c$]{};
  \node[place] (p_4'') at (-3.5*\b, -4.5*\c) [label=right:$p_4$]{};
  \node[place] (p_2'') at (-4.5*\b, -4.5*\c) [label=left:$p_2$]{};
  \node[transition] (b'') at (-4.5*\b, -5.5*\c) [label=right:$b$]{};
  \node[place] (p_3'') at (-4.5*\b, -6.5*\c) [label=left:$p_3$]{};
  \path[->] (p_1) edge (c'') (c'') edge (p_4'') edge (p_2'')
  		(p_2) edge (c'')
		(p_2'') edge (b'') (b'') edge (p_3'')
	;

\end{tikzpicture}}
  \caption{A safe Petri net (left) and one of its branching processes (right).
    Configurations \(\minconf{e}\) and \(\minconf{e'}\) lead to the
    same marking \(\{p_1, p_3\}\) and have isomorphic extensions (in gray). The
    dashed arrow represents the fact that \(\minconf{e'} \Adq \minconf{e}\).
    Consequently \(e\) is a cut-off.}
  \label{fig:pb_cut-off}
\end{figure}
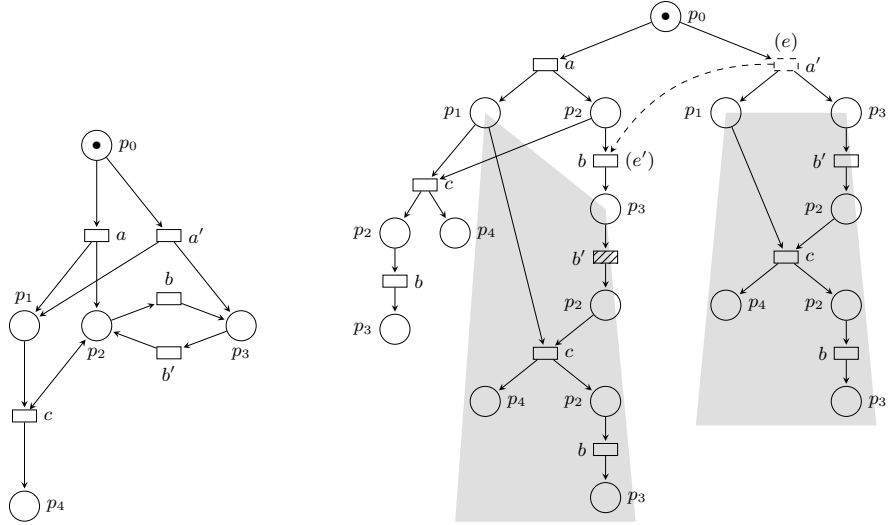

Fig.~\ref{fig:pb_cut-off} illustrates this situation. Let the goal be \(\goal
= \{p_4, p_3\}\). It can be reached by the firing sequences \(a(bb')^*c(bb')^*b\)
or \(a'b'(bb')^*c(bb')^*b\). Only those who do not take the cycle \(bb'\) are
minimal, namely \(acb\) and \(a'b'cb\). Notice that all the transitions
participate in at least one minimal firing sequence, so the model \(\model\)
cannot be reduced from the initial marking (every reduction procedure will
output \(\useless(\model, \goal) = \emptyset\)). On the other hand, if
transition \(a\) is fired, we reach marking \(\{p_1, p_2\}\) from which \(b'\),
\(a\) and \(a'\) become useless.

Now, observe the branching process on the right of Fig.~\ref{fig:pb_cut-off}
(it is a prefix of the unfolding \(\unfolding\) of \(\model\)). Notice that the
causal past \(\minconf{e}\) of the event labeled \(a'\) and the causal past of
the event \(e'\) labeled \(b\) reach the same marking \(\marking{\minconf{e}} =
\marking{\minconf{e'}} = \{p_1, p_3\}\). Moreover, an adequate order on the
configurations of \(\unfolding\) may order them as \(\minconf{e'} \Adq
\minconf{e}\). Consequently, \(e\) is a cut-off and the minimal configuration
\(\Config(a'b'cb)\) is not represented in the finite prefix. Following the idea of
the proof of completeness of finite prefixes based on adequate orders,
we can indeed shift the extension \(b'cb\) of \(\minconf{e}\) (in gray on the
right of Fig.~\ref{fig:pb_cut-off}) to the isomorphic extension of
\(\minconf{e'}\) (also in gray on the figure). We get the configuration
\(\Config(abb'cb)\), which reaches the goal as well. But this configuration is
\emph{not} minimal any more because it executes the cycle \(bb'\): the marking
reached after \(a\) is the same as the marking reached after \(abb'\). Actually,
the model reduction procedure called from the event labeled \(a\) may very well
have declared \(b'\) useless.
Consequently, \(\Config(abb'cb)\) would not be represented in the prefix.
We correct this by allowing after \(\minconf{e'}\) all the transitions that were
allowed after \(\minconf{e}\).


\medskip


The difficulty in the definition and in the computation of a finite prefix
\(\prefix\sgd\) of \(\unfolding\) which preserves the markings reachable
in $\unfolding\sgd$
is to allow in the future of an event \(e'\) all the transitions that are useful
for at least one of all the configurations which are shifted to \(\minconf{e'}\)
by the mechanics described above.
The first answer to this problem is to allow after \(\minconf{e'}\) all the
transitions that were allowed after \(\minconf{e}\). This solves the problem of
an event consuming only post-conditions of \(e'\), like the occurrence of \(b'\)
after \(e\) in our example of Fig.~\ref{fig:pb_cut-off}: its corresponding event
after \(e'\) is now allowed. However, this is not sufficient in general:
an event \(f\) consuming a post-condition of \(e\) may also consume
other conditions which are created by events concurrent to \(\minconf{e}\). Such
event \(f\) has a corresponding \(f'\) in the future of \(e'\), consuming
conditions which are available after firing a configuration of the form
\(\minconf{e'} \cup \config'\) for some \(\config'\) concurrent to
\(\minconf{e'}\). 
We need transition \(t = h(e) = h(e')\) to be allowed after
all the conditions consumed by \(f'\). In the case of a condition \(c' \in
\preset{f'} \setminus \cut{\minconf{e'}}\), our procedure ensures this as
follows: if it calls the model reduction procedure after the event
\(\preset{c'}\),
it also calls it on the marking
\(\marking{\minconf{e'} \cup \minconf{\preset{c'}}}\) which equals
\(\marking{\minconf{e} \cup \minconf{\preset{c}}}\). Hence, if \(t\) is needed
after \(\minconf{e} \cup \minconf{\preset{c}}\), it will also be allowed after
\(c'\).
In the end, when applying the reduction procedure
after a configuration \(\config\), we also take into account a set
\(\altconfig{\config}\) of alternating configurations defined
inductively as:
\begin{itemize}
\item \(\config \in \altconfig{\config}\)
\item \(\forall \config' \in \altconfig{\config}\), 
  \(\forall e, e' \in E\) such that \(\minconf{e} \rhd \minconf{e'}\)
  and \(\marking{\minconf{e}} = \marking{\minconf{e'}}\),\\
  \indent if \(\cut{\minconf{e'}} \cap \postset{\config'} \neq \emptyset\)
  and \(\minconf{e'} \cup \config'\) is conflict free, then 
  \(\minconf{e'} \cup \config' \in \altconfig{\config}\).
\end{itemize}

However, in practice, during the computation of the goal-driven prefix,
\(\altconfig{\config}\) will be computed on the events and configurations
derived so far, hence ignoring events later added in the prefix. Also, as
explained above, when an event \(e\) is stated cut-off because of a
$\Adq$-smaller event $e'$, we allow after \(e'\) all the transitions allowed
after \(e\); but this implies reconsidering some new extensions of \(e'\).

For these reasons, the procedure \(\Call{GD-Prefix}{\petrinet}\) presented in
Algorithm~\ref{alg:gdprefix} iterates the computation of a putative
prefix,
progressively refining an over-approximation of
transitions to ignore (map $\Vignored$), by identifying \textit{a posteriori} the
transitions that should \emph{not} have been ignored.
\medskip

At each iteration, the procedure
$\Call{Putative-GD-Prefix}{\petrinet,\Vignored}$ computes a
putative prefix, relying on the previous value of the map $\Vignored$ of
transitions that can be ignored. Essentially, the prefix \(\prefix\) obtained at
the first iteration 
is the naive prefix
of \(\unfolding_\mathrm{gd}\) (prefix without the gray parts on the example of
Fig.~\ref{fig:pb_cut-off}).

Once a putative prefix has been computed, we verify \emph{a posteriori} if its
related map $\Vignored$ is correct.
This is done by re-computing $\Vignored$ using the procedure
\Call{Post-$\Vignored$}{$\Vignored,\prefix$}, this time taking into account all
the events in $\prefix$ (line~\ref{line:recompute}).
By construction, the resulting $\Vignored'$ can only allow more transitions than
\(\Vignored\). If $\Vignored'$ differs from $\Vignored$,
a new putative prefix is computed according to the corrected $\Vignored'$.

The procedure \(\Call{Post-$\Vignored$}{\Vignored,\prefix}\) takes the
\(\altconfig{\minconf{e}}\) into account by the way of a modified version of
\(\ignored\), now defined on conditions rather than on events.
Given a condition $c\in\postset e$ in a prefix $\prefix$,
\[\begin{array}{@{}l@{}}
\ignored{c,\Vignored,\prefix} \eqdef\\\qquad \left\{
\begin{array}{@{}l@{\quad}l@{}}
  \bigcup_{c' \in \preset{e}} \Vignored(c') & \mbox{if \(e \notin E'\)} \\
\bigcap_{C^*\in\altconfig{\minconf e}}
  \uselessmg
      {\marking{C^*}}
      {\bigcup_{c' \in \preset{e}} \Vignored(c')} & \mbox{if \(e \in E'\),}
\end{array}\right.
\end{array}
\]
where $E'$ is the set of events triggering an explicit reduction
(Section~\ref{sec:gd-unf}).

This iterative construction necessarily terminates (Lemma~\ref{lem:termination},
proof in Appendix~\ref{proof:termination})
and converges to a unique finite prefix $\prefix\sgd$.
Regarding complexity, putting aside the call to model reduction,
whereas all the structures are finite,
$\altconfig{\config}$ can have an exponential numbers of configurations due to
multiple combinations of configurations sharing an intersection.

\begin{lemma}
  The procedure \Call{GD-Prefix}{$\petrinet$} terminates.
  \label{lem:termination}
\end{lemma}

Notice that
$\prefix\sgd$ may contain events that are not in
\(E_\mathrm{gd}\).
Hence,
goal-driven prefix is a prefix of \(\unfolding\), but not necessarily a prefix
of \(\unfolding_\mathrm{gd}\). This is the case of the event labeled \(b'\)
after \(e'\), as we discussed above for the example in Fig.~\ref{fig:pb_cut-off}.

Theorem~\ref{th:prefix} (proof in Appendix~\ref{proof:prefix})
states completeness of $\prefix\sgd$ w.r.t.\ minimal configurations.
Thus, the goal-driven prefix preserves the reachability of the goal.
One can finally remark that, by construction, $\prefix\sgd$ contains at most one non-cutoff event
per reachable marking, assuming the adequate order $\Adq$ is total.%

\begin{theorem}\label{th:prefix}
  For every configuration \(\config\) of \(\unfolding_\mathrm{gd}\)
  and for every single-event extension \(\{f\}\) of \(\config\) such that
  \(\config \cup \{f\}\) is a prefix of a minimal configuration to the goal,
  there exists a configuration \(\config'\) 
  in the goal-driven
  prefix and a single-event extension \(\{f'\}\) of \(\config'\) with
  \(\marking{\config} = \marking{\config'}\) and \(h(f) = h(f')\).
\end{theorem}


\paragraph{Example.}
Let us consider the Petri net of Fig.~\ref{fig:pb_cut-off}(left) with the goal $\{p_4,p_3\}$.

The goal-driven unfolding can lead to the branching process of
Fig.~\ref{fig:pb_cut-off}(right) where
the dashed transition $b'$ has been removed. Indeed, after transition $a$, transition $b'$ is
declared useless as it is not part of any minimal configuration extending $\Config(a)$.
Therefore 3 maximal configurations are remaining in the goal-driven unfolding:
the two minimal configurations $\Config(acb)$ and $\Config(a'b'cb)$, and the configuration
$\Config(ab)$ which does not reach the goal.

The goal-driven prefix can lead to the branching process of Fig.~\ref{fig:pb_cut-off}(right) where
the event $e$ is cut-off (because of $e'$), and therefore its future events are ignored,
and where one of the two remaining events firing transition $c$ is declared cut-off (because of the other
one).
Although the $b'$ transition can be declared useless after $\Config(a)$ (and hence $\Config(ab)$),
the cut-off of $e$ will remove $b'$ from the set of ignored transitions of the conditions matching
with $p_1$ and $p_3$ on the cut of $\Config(ab)$.
Therefore, the events and conditions in the left gray area will be added to the prefix, from which
all the minimal configurations can be identified.


\section{Experiments}
\label{sec:experiments}

%

In this section, we compare the size of the complete prefix with the goal-driven prefix
on different Petri net models of biological signalling and gene regulatory networks.
In general, such networks gather dozens to thousands nodes having sparse
interactions (each node is directly influenced by a few other nodes), which
call for concurrency-aware approaches to cope with the state space explosion.
We took the networks from systems biology literature, specified as Boolean or automata networks:
each node is modelled by an automaton, where states model its
activity level, most often being binary (active or inactive).
The Petri nets are encodings of these automata networks which ensure bisimilarity \cite{CHJPS14-CMSB}.%

\paragraph*{Implementation.}
In practice, 
instead of computing putative prefixes from scratch
as it is described in Algorithm~\ref{alg:gdprefix},
our implementation for the goal-driven prefix\footnote{Code and models available at 
\url{http://loicpauleve.name/godunf.tbz2}}
iteratively corrects the putative prefix by propagating transitions missed in
the previous iteration.
At this stage, it does not use any particular optimization \cite{BBCKRS12},
our primary objective being to compare the size of the resulting prefixes.
In order to obtain a proper comparison \cite{KKV03}, our implementation uses the same
arbitrarily-fixed ordering for the complete and goal-driven prefixes extensions.

The computation of $\uselessmg M I$ relies on the goal-oriented reduction
of asynchronous automata networks introduced in \cite{gored}.
This method is based on a static analysis of causal dependencies of transitions and an
abstract interpretation of traces which allow to collect all the transitions
involved in the minimal configurations to the goal:
non-collected transitions can then be ignored.
The complexity of the reduction is polynomial with the number of automata
and transitions, and exponential with the number of states in individual automata
(i.e., number of qualitative states of nodes).
As shown in \cite{gored}, the method can lead to drastic model reductions
and can be executed in a few hundredths of a second on networks with several hundreds of nodes.

We applied the goal-driven unfolding to 1-safe Petri net encodings of the
automata networks, where there is one place for each local state of each
automaton, and a one-to-one relationship between transitions.
The places corresponding to states of a same automaton are mutually
exclusive by construction.
Future work may consider goal-driven unfolding of products of transition systems
\cite{Esparza-CONCUR99}.

The goal-driven prefix we define in this paper supports calling the model
reduction procedure at discretion:
even if it has a low computational cost,
performing the model reduction after each event may turn out to be very time consuming.
Our prototype implements simple strategies to decide when the call to the model reduction should
be performed: after each event; only for the first $n$ events; and only for events up
to a given level in the unfolding.

\paragraph*{Benchmarks.}
Given a Petri net with an initial marking $M_0$ and a goal $g$,
we first compute the goal-oriented model reduction from initial marking
($\uselessmg {M_0}{\emptyset}$).
The resulting net is then given as input to the unfolding, either with the
complete finite prefix computation, or with the goal-driven.
Therefore, the difference in the size of the prefixes obtained is due only to
transition exclusions after at least one event.

Table~\ref{tab:experiments} summarizes the benchmarks between complete and goal-driven prefix
on different models of biological networks.
The size of a prefix is the number of its non-cutoff events.
\begin{table}[tb]
\centering
\begin{tabular}{|l|l|l|r|r|r|}
\hline
Model & Prefix & Strategy & Prefix size & Time & Nb reductions\\
\hline
RB/E2F &
	complete & N/A & 15,210 & 24s & N/A\\
$|P|=80 \quad |T|=54$ &
	goal-driven & always & 112 & 0.5s & 136\\
\hline
T-LGL &
	complete & N/A & $>$1,900,000$^*$ & OT$^*$ & N/A\\
$|P|=98\quad |T|=159$ &
	goal-driven & always & 17 & 0.3s & 17\\
\hline
VPC &
	complete & N/A & 44,500 & 176s & N/A\\
$|P|=135\quad |T|=216$ &
	goal-driven & always & 1,827 & 2h & 16,009\\
&	& first $1,000$ & 2,036 & 60s & 1,000\\
&   & level $\leq 2$ & 2,400 & 7s & 38\\
\hline
\end{tabular}
\caption{Benchmarks of the goal-driven w.r.t. complete prefix of 1-safe Petri nets.
For each model, the number of places $|P|$ and transitions $|T|$ is given.
The strategy decides when the model reduction should be performed; the number of
calls to the reduction procedure is indicated in the column ``Nb reductions''.
Computation times were obtained on an Intel\textregistered{} Core\texttrademark{} i7 3.4GHz CPU
with 16GB RAM.
N/A: Non Applicable;
$^*$: out-of-memory computation (with mole \cite{mole}, with the same ordering for extensions as our implementation), the indicated prefix size is only a lower bound.}
\label{tab:experiments}
\end{table}
``RB/E2F'' is a model of the cell cycle \cite{e2frb};
``T-LGL'' is a model of survival signaling in large granular lymphocyte leukemia \cite{tlcg};
and ``VPC'' is a model for the specification of vulval precursor cells and cell fusion
control in Caenorhabditis elegans \cite{WM13-FG}.
All those models have very different network topology and dynamical features.
For each model, the initial marking and goal correspond to biological states of interest
(checkpoints or differentiated states).

On these models, the goal-driven prefix shows a significant size reduction, while
containing all the minimal configurations.
The number of reductions can be larger than the size of the prefix as it accounts for the
intermediate putative prefixes (as explained in Section~\ref{sec:gd-prefix}).
For the ``VPC'' model, we applied several strategies for deciding when the model reduction should be
called.
In this case, the systematic model reduction led to some re-ordering of the extensions and cut-offs
declaration, which required numerous additional calls to the model reduction procedure.
This motivates the design of heuristics to estimate when a model reduction should be performed.
For the ``T-LGL'' model, it was impossible to compute the complete finite prefix, whereas the
goal-driven cuts most of the configurations and produces a very concise prefix.
This behaviour can be explained by large transient cycles prior to the goal
reachability, which are avoided by the use of model reduction during the prefix computation.


\section{Conclusion}
\label{sec:conclusion}
We introduced the goal-driven unfolding of safe Petri nets for identifying efficiently \emph{all} the minimal
configurations that lead to a given goal.
The goal can be a marking of the net, or any partially specified marking, and notably a single
marked place.
The goal-driven unfolding relies on an external reduction method which identifies transitions that
are not part of minimal configuration for the goal reachability.
Such useless transitions are then skipped by the unfolding.
The computation of a goal-driven prefix requires a particular treatment of
cut-offs to ensure that all the markings reachable in the goal-driven unfolding are
preserved.

We applied our approach to different models of biological systems which show
a significant reduction of the prefix when driven by the goal.
In our framework, the reduction procedure can be applied at discretion, and many possible heuristics
could be embedded to decide when the reduction is timely, which impacts both the execution time and
the size of the prefix.
The resulting goal-driven prefix contains fewer events prefix than reachable markings, due to the
total adequate order, as well as for classical finite complete prefix.

Future work will explore the combination with the semi-adequate ordering of configurations
of directed unfolding \cite{BonetHHT08} as it may reduce the need for propagating transitions
allowed by a cut-off event.
Although our approach considers the model reduction procedure as a blackbox,
on-going work is currently generalizing the one used in the experimentations to any safe Petri net.
Finally, we are considering implementing the goal-driven unfolding within Mole~\cite{mole}.

\bibliographystyle{abbrv}
\bibliography{godunf}

\newpage
\appendix
\section{Proof of Theorem~\ref{th:godunf}}\label{proof:godunf}

Let \(E\) be a minimal configuration of \(\model\) to the goal. We want to prove
that \(E \subseteq E_{\mathrm{gd}}\). Given the definition of
\(E_{\mathrm{gd}}\) and given that \(E\) is causally closed, this is equivalent
to proving that \(E \cap E_\mi{Ignored} = \emptyset\}\): trivially, if some \(e
\in E\) is also in \(E_\mi{Ignored}\), then \(\minconf{e} \cap E_\mi{Ignored}\)
contains at least \(e\), so it is nonempty, and by definition \(e \notin
E_{\mathrm{gd}}\); conversely, if some \(e \in E\) is not in
\(E_{\mathrm{gd}}\), this is because \(\minconf{e} \cap E_\mi{Ignored} \neq
\emptyset\), and since \(E\) is causally closed, \(\minconf{e} \subseteq E\),
which implies that \(E \cap E_\mi{Ignored} \neq \emptyset\).

Then, it remains to show that \(E \cap E_\mi{Ignored} = \emptyset\). Let
\(e = \tup{C, t} \in E\);
we have to show that, for every \(e' \in \preconf{e}\), \(t \notin
\ignored{e'}\). If \(e' \notin E'\), then, by definition, \(\ignored{e'} \eqdef
\bigcup_{e'' \in \preconf{e'}} \ignored{e''}\), which means that \(e'\) has a
causal predecessor \(e''\) which also satisfies \(t \notin \ignored{e''}\).

Select now an \(e'\) which is minimal w.r.t.\ causality. This eliminates the
previous case, so we have \(e' \in E'\) and
\(t \in \bigcup_{e'' \in \preconf{e'}} \ignored{e''}\) and
\(t \notin \useless\Big(\model, \goal, \marking{\minconf{e'}}, \bigcup_{e'' \in
    \preconf{e'}} \ignored{e''}\Big)\).
Assuming that \(\useless\) is a reduction procedure satisfying
Definition~\ref{def:reduction}, this implies that no minimal firing sequence
from \(\marking{\minconf{e'}}\) to the goal uses \(t\), which contradicts the
fact that \(E\) is a minimal configuration to the goal: Indeed,
since \(e' \in \preconf{e}\), there exists a
linearization\footnote{A linearization of \(E\) is a total ordering of \(e_1,
  \dots, e_{|E|}\) of the events in \(E\) such that for \(e_i < e_j \implies i <
  j\).}
\(e_1, \dots, e_{|E|}\) of \(E\) in which the events in
\(\minconf{e'}\) occur before the others, i.e.\ \(\{e_1, \dots,
e_{|\minconf{e'}| - 1}\} = \preconf{e'}\), \(e_{|\minconf{e'}|} = e'\) and
\(\{e_{|\minconf{e'}| + 1}, \dots, e_{|E|}\} = E \setminus \minconf{e'}\);
then \(t_{|\minconf{e'}| + 1} \dots t_{|E|}\) (with \(t_i \eqdef h(e_i)\) the
transition corresponding to \(e_i\))
is a firing sequence from \(\marking{\minconf e'}\) to the goal and it uses
\(t\). If it is not minimal, then because it has a feasible cycling
permutation \(\sigma\), then \(t_1, \dots, t_{|\minconf{e'}|} \cdot \sigma\)
is a feasible cycling permutation from \(M_0\) to the goal, which
contradicts the fact that \(E\) is a minimal configuration to the goal.
\qed

\section{Proof of Lemma~\ref{lem:termination}}\label{proof:termination}
Because the set of markings is finite and because
$\altconfig{\config}$ is also finite (computed on a finite prefix),
procedure \Call{Putative-GD-Prefix}{$\petrinet,\Vignored$} always terminates;
moreover all the iterations in procedure
\Call{Post-$\Vignored$}{$\Vignored,\prefix$} are over finite sets.
Finally, we prove that procedure \Call{GD-Prefix}{$\petrinet$}
terminates, i.e., after a finite number of iterations,
\Call{Post-$\Vignored$}{$\Vignored,\prefix$} $=\Vignored$.
First, by construction,
$\forall c\in\Vignored$, $c\in\Vignored'$ and $\Vignored'(c) \subseteq \Vignored(c)$,
with $\Vignored'=$ \Call{Post-$\Vignored$}{$\Vignored,\prefix$}.
Then, remark that, due to the cut-off treatment,
any event of any putative prefix has a bounded number of event ancestors (causal past):
the number of reachable markings.
Finally, because the branching up to a given depth is finite, only a finite number of events
can be considered in any iteration of the putative prefix;
hence the number of events registered in $\Vignored$ is finite.
Therefore, due to the monotonicity of $\Vignored$ modifications,
the iterative procedure necessarily converges towards a unique finite prefix in a finite number of
steps.
\qed

\section{Proof of Theorem~\ref{th:prefix}}\label{proof:prefix}



  We first show that for every configuration \(\config\) that can be extended to
  a minimal configuration to the goal, 
  there exists a configuration in the goal-driven prefix which contains no
  cut-off event and reaches \(\marking{\config}\).
  The principle is the one used for completeness of classical finite prefixes
  defined using adequate orders: if \(\config\) contains no cut-off event, it is
  in the goal-driven prefix, since the construction of the \(\mathit{Useless}\)
  is more permissive in the goal-driven prefix (with the use of \(\altconfig\))
  than in the goal-driven unfolding. Now, if \(\config\) contains a cut-off
  event \(e\) (w.r.t.\ an event \(e'\) such that \(\minconf{e'} \Adq
  \minconf{e}\) and \(\marking{\minconf{e'}} = \marking{\minconf{e}}\)), then
  \(\config\) can be decomposed as \(\minconf{e} \uplus D\) and \(e'\) has an
  extension \(D'\) isomorphic to \(D\). Then \(\minconf{e'} \uplus D'\) is
  smaller than \(\config\) w.r.t.\ \(\Adq\) and reaches the same marking.
  This operation can be iterated if needed; it terminates because \(\Adq\) is well
  founded, and gives a configuration \(\config'\) without cut-offs which reaches
  the same marking as \(\config\). If \(\config\) can be extended with an event
  \(f\), then so can \(\config'\) with an event \(f'\) corresponding to the same
  transition \(h(f') = h(f)\). The event \(f'\) is in the prefix but may be a
  cut-off.

  It remains to make sure that the transitions of \(\config\) (plus \(h(f)\))
  are not considered useless. For this, focus on \(\minconf{e} \uplus D\) mapped
  to \(\minconf{e'} \uplus D'\). For every event \(d \in D\), let \(d'\) be the
  corresponding event in \(D'\). We have \(\minconf{d} \cup \minconf{e} \in
  \altconfig{\minconf{d'}}\) because the causal past of \(d'\) uses at least one
  condition from the cut of \(\minconf{e'}\). This ensures that the transitions
  fired in \(D\) after \(\minconf{d} \cup \minconf{e}\) are taken into account
  in the computation of the transitions allowed after \(d'\) (if \(d'\) is
  itself in the prefix, otherwise apply this inductively), ensuring that in
  the end \(\config' \cup \{f'\}\) is in the goal-driven prefix.
  \qed

\end{document}